\documentclass[journal]{IEEEtran}
\IEEEoverridecommandlockouts
\usepackage{amsfonts,amsmath,amssymb}
\usepackage{cite}
\usepackage{graphicx}
\usepackage{url}
\usepackage{bm}
\newtheorem{theorem}{Theorem}
\usepackage{color}
 \newtheorem{lemma}{Lemma}
 \newtheorem{definition}{Definition}
  \newtheorem{remark}{Remark}

\ifCLASSINFOpdf
\else
\fi
\usepackage{lineno}

\begin{document}
%

\title{On Designs of Full Diversity Space-Time Block Codes   for Two-User MIMO Interference Channels}

\author{Long~Shi,~\IEEEmembership{Student Member,~IEEE}, ~Wei~Zhang,~\IEEEmembership{Senior Member,~IEEE},~
        and ~Xiang-Gen~Xia,~\IEEEmembership{Fellow,~IEEE}
%

\thanks{L. Shi and W. Zhang are with School of Electrical Engineering and Telecommunications,   The University of New South Wales,
  Sydney, Australia (e-mail: long.shi@student.unsw.edu.au; wzhang@ee.unsw.edu.au). Their work was supported in part by the Australian Research Council Discovery Project DP1094194. Part of this work was
presented at the IEEE Global Communications Conference
(GLOBECOM) -- Broadband Wireless Access Workshop, Anaheim, CA, USA, Dec. 3, 2012.}
   \thanks{X.-G. Xia is with Department of Electrical and Computer Engineering, University of Delaware, DE 19716, USA (e-mail: xxia@ee.udel.edu), also with the Chonbuk National University, Jeonju, South Korea. His work was supported in part by the National Science Foundation (NSF)
under Grant CCF-0964500, the Air Force Office of Scientific Research (AFOSR) under Grant FA9550-12-1-0055, and the World Class University (WCU) Program, National Research Foundation, South Korea.
}
}
\maketitle


\begin{abstract}
In this paper, a design criterion for space-time block codes (STBC) is proposed for two-user MIMO interference channels  when a group zero-forcing (ZF) algorithm is applied at each receiver to eliminate the inter-user interference. Based on the design criterion, a design of STBC for two-user interference channels is proposed that can achieve full diversity for each user with the group ZF receiver. The code rate approaches one  when the time delay in the encoding (or code block size) gets large.    Performance results demonstrate that the full diversity can be guaranteed by our proposed STBC with the group ZF receiver.


\end{abstract}

\begin{IEEEkeywords}
 Two-user MIMO interference channels,  space-time block codes, full diversity, group ZF receiver.
\end{IEEEkeywords}

\section{Introduction}

 It is undoubted that how to maximize diversity  gain  is one of the major concerns in both point-to-point communication systems and wireless networks.  As for point-to-point communication systems, remarkable progress has been witnessed in the field of  STBC designs to achieve maximum diversity gain over the last decade. In particular, some pioneering works are well known of the full-diversity criteria based on different types of receivers such as maximum likelihood (ML) and linear receivers \cite{tar,liu1,liu2,shang,guo,guocor}.

 Recently, the contributions aforementioned  shed  light on the STBC designs   for wireless networks such as multiple access channels (MACs) \cite{nag,song,kaz09, quan,feng1,feng2} and X channels \cite{liang}.
For two-user MACs,  various cancellation schemes  in  \cite{nag,song,kaz09,quan} were considered with the help of Alamouti coding to suppress the interference from the neighboring user. In \cite{nag}, a minimum mean squared error (MMSE) interference cancellation was developed  when   two receive antennas are equipped at each terminal. Later, a family of   interference cancellation schemes based on Bayesian analysis was proposed in \cite{song}. Then,  \cite{kaz09} analyzed the diversity order of Alamouti and quasi-orthogonal STBC for multiuser MAC with two decoding algorithms, i.e., joint ML decoding and group ZF decoding. It showed that for joint ML decoding full diversity is obtained and for the group ZF decoding full diversity cannot be achieved by Alamouti and quasi-orthogonal STBC.  In order to
obtain an optimal power gain, the recent work in \cite{quan} employed  ZF receiver to eliminate the interference.
  Unfortunately, all the schemes in  \cite{nag,quan,kaz09,song}   cannot guarantee   the maximum possible diversity order after interference cancellation. To achieve full diversity in MACs, the works in \cite{feng1}\cite{feng2}   developed a particular precoder  for each user with arbitrary number of transmit antenna. Interference alignment was firstly proposed to suppress the interference and achieve the maximum degree of freedom (DoF) in MIMO interference and X channels \cite{jafar1, jafar2, jafar3}.
For two-user X channels,  a  full-diversity transmission scheme in \cite{liang} is proposed by borrowing an idea from interference alignment in \cite{jafar1}. Different from the work in \cite{jafar1}, the goal of \cite{liang} is to achieve the maximum  diversity order instead of the maximum DoF in   X channels.  Note that  the inter-user interference was safely eliminated  in \cite{ feng1, feng2,  liang}, since it can be aligned into an orthogonal subspace to the desired signals with the aid of full channel state information at the transmitters (CSIT). By doing so, interference cancellation is realized   without any side impact on the full diversity   of the desired user. In practice, however, it is very hard for the transmitters to fully know the channel information. Therefore, it is  natural to consider   a full-diversity STBC design for wireless networks without CSIT.

In this paper, we propose a  design criterion for an STBC to achieve full diversity  in two-user MIMO interference channels without  CSIT, where the full diversity is defined for a single user pair, i.e., the link between a transmitter  and its corresponding receiver. Since the receiver corresponding to the desired user knows nothing about the codebook of the interfering user,    a straightforward way  to eliminate the interference  is using  the group ZF receiver, which was widely investigated and applied in space-time coded MIMO systems \cite{guo}\cite{guocor} \cite{wei, shi, xucon}.   Based on the  design criterion, we propose a systematic STBC design for each user with multiple transmit antennas that can obtain full diversity after interference cancellation. For a fixed number of transmit antennas $M$, the code rate  approaches one  when the time delay in the encoding (or code block size) gets large.  
  In our paper, we focus on two-user interference channels where the codebooks of two users are not shared. Hence, the joint ML decoding cannot be used.   Compared to the work in \cite{nag,quan,kaz09,song}, our  contribution lies in the design of full diversity STBC for interference channels with the group ZF decoding.
It is worthwhile to mention that the main difference  from the works in  \cite{ feng1, feng2,  liang} is that the channel information is unknown at the transmitter side  such that in our paper the inter-user interference is canceled by using the group ZF receiver instead of interference alignment. To our best knowledge, this is the first attempt of full diversity STBC
design for two-user MIMO interference channels  without CSIT.

The rest of the paper is organized as follows. The system model  is outlined in Section II.  Then, a design criterion of STBC under the group ZF receiver is given  in Section III. In Section IV, a systematic STBC design is proposed for each user  and  the full diversity is proved when   the group ZF receiver is utilized.  Simulation results are presented  in Section V. Finally, we conclude the paper.

\emph{Notations}:  Superscripts  ${}^{\mathcal{T}}$ and ${}^\dag$ stand for transpose and conjugate transpose, respectively.  $\mathbb{C}$ denotes the  complex number field.
$\mathbf{I}_n$ denotes an $n\times n$ identity matrix and  $\mathbf{0}_{m\times n}$ denotes an $m\times n$ matrix  whose elements are all $0$.    ${\rm diag}\{\mathbf{v}\}$ produces a diagonal matrix with  entries of a vector $\mathbf{v}$ in main diagonal. ${\rm span}\{\mathbf{v}_1, \mathbf{v}_2, \ldots, \mathbf{v}_M\}$ denotes a subspace spanned by all the vectors of $\{\mathbf{v}_1, \mathbf{v}_2, \ldots, \mathbf{v}_M\}$. $E\{\}$ denotes the statistical expectation.  $ \otimes$ denotes Kronecker product.  $\| \ \| $   represents the Frobenius norm.  ${\rm tr \{\mathbf{A}\}}$ denotes the trace of a square matrix $\mathbf{A}$. $\min\{\mathbf{v}\}$ denotes the minimum value in a vector $\mathbf{v}$.

\section{System Model }

Fig. \ref{fig:1} shows  a two-user interference channel composed of $2$ transmitters (i.e.,  {user 1} and  {user 2}) and $2$ receivers (i.e.,  {R1} and  {R2}), where each user is equipped with $M$ transmit antennas and its corresponding receiver has $N$ receive antennas.

 \begin{figure}[t]
 \centering
        \includegraphics[width=0.85\columnwidth]{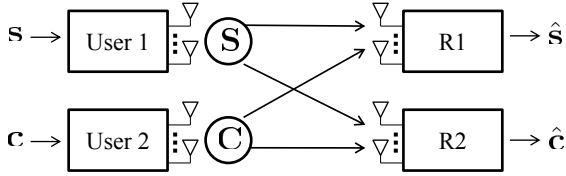}
       \caption{System model of two-user MIMO interference channel.}
        \label{fig:1}
\end{figure}

At the transmitter side, two symbol vectors $\mathbf{s}=\left[s_1, s_2,\ldots,s_L\right]^\mathcal{T}$ and $\mathbf{c}=\left[c_1, c_2,\ldots, c_L\right]^\mathcal{T}$ are firstly encoded to   STBCs   $\mathbf{S} \in \mathbb{C}^{T\times M}$ and $\mathbf{C}\in \mathbb{C}^{T\times M}$, respectively. The independent symbols $s_l$ and $c_l$ for $l=1,2,\ldots,L$ are selected from two independent codebooks $\mathcal{A}_1$ and $\mathcal{A}_2$ such as quadrature amplitude modulation (QAM). Then, the codewords $\mathbf{S}$ and $\mathbf{C}$ are transmitted from user 1 for R1 and user 2 for R2 over $T$ time slots, respectively.

The received signals $\mathbf{Y}_1 \in \mathbb{C}^{T\times N}$ and $\mathbf{Y}_2\in \mathbb{C}^{T\times N}$ at R1 and R2 are given by, respectively
\begin{eqnarray}\label{Y1ic}
\mathbf{Y}_1&=&\sqrt\frac{\rho}{\mu}\mathbf{S}\mathbf{H}_1+
\sqrt\frac{\rho}{\mu}\mathbf{C}\mathbf{G}_1+\mathbf{N}_1\\
\label{Y2ic} {\rm and}~ \mathbf{Y}_2&=&\sqrt\frac{\rho}{\mu}\mathbf{S}\mathbf{H}_2+
\sqrt\frac{\rho}{\mu}\mathbf{C}\mathbf{G}_2+\mathbf{N}_2,
    \end{eqnarray}
where    the matrices $\mathbf{H}_r$ and $\mathbf{G}_r$  stand for the channel matrices from user 1 to R$r$ and user 2 to R$r$ for $r=1,2$, respectively.  Moreover, the   entries of the channel matrices $\mathbf{H}_r \in \mathbb{C}^{M\times N}$ and $\mathbf{G}_r \in \mathbb{C}^{M\times N}$ for $r=1,2$ are assumed to be independently and identically distributed (i.i.d.) complex Gaussian with zero mean and unit variance, i.e., $\mathcal{CN}(0,1)$, and the channels are assumed to experience the quasi-static fading.
$\mathbf{N}_r \in \mathbb{C}^{T\times N}$  is the noise matrix at R$r$ for $r=1,2$, whose elements are also i.i.d.   $\mathcal{CN}(0,1)$.  Additionally, $\rho$ is the average signal-to-noise ratio (SNR) per receive antenna at each receiver and   $\mu$
is the normalization factor to ensure that the average energy of the coded symbols transmitting from all antennas of each user during one symbol period is 1.

Throughout this paper, there are two assumptions  as follows
  \begin{enumerate}
    \item The channel matrices $\mathbf{H}_1$ and $\mathbf{G}_1$ are perfectly known  at R1, and $\mathbf{H}_2$ and $\mathbf{G}_2$ are known at R2. However, no channel information is known at any transmitter.
    \item The codebook of each user is only known at its respective receiver.
  \end{enumerate}

 To decode the desired information at each receiver,  (\ref{Y1ic}) and (\ref{Y2ic}) can be rewritten as the equivalent forms by using some linear transformations, respectively
\begin{eqnarray}\label{Y1ice}
\mathbf{y}_1&=&\sqrt\frac{\rho}{\mu}{\bm{\mathcal{H}}}_1 {\mathbf{s}}
+\sqrt\frac{\rho}{\mu}\bm{\mathcal{G}}_1 {\mathbf{c}}+\mathbf{n}_1\\
\label{Y2ice} {\rm and}~ \mathbf{y}_2&=&\sqrt\frac{\rho}{\mu}\bm{\mathcal{H}}_2 {\mathbf{s}}
+\sqrt\frac{\rho}{\mu}\bm{\mathcal{G}}_2 {\mathbf{c}}+\mathbf{n}_2,
    \end{eqnarray}
 where $\mathbf{y}_r \in \mathbb{C}^{TN\times 1}$  is the equivalent form of $\mathbf{Y}_r$ for $r=1,2$. Likewise, $\mathbf{N}_r$  is equivalently transformed into $\mathbf{n}_r$ for $r=1,2$. Moreover,  the equivalent channel matrices are given by $\bm{\mathcal{H}}_r\in\mathbb{C}^{TN\times L}$ and $\bm{\mathcal{G}}_r \in\mathbb{C}^{TN\times L}$   for $r=1,2$,  both of which   are determined by the code matrices  $\mathbf{S}$ and $\mathbf{C}$  and the original channel matrices $\mathbf{H}_r$ and $\mathbf{G}_r$, respectively.

From now on, we focus on the interference cancellation and decoding at R1. The process at R2 follows similarly and is omitted.

 As we mentioned,   the desired information at R1 is  interfered by the  information of user 2.  The group ZF receiver is used  to cancel the interference and decode the desired information of the respective user.
Firstly, the cancellation matrix of  the group ZF receiver at R1 is written as
 \begin{eqnarray}\label{Qcan}
\mathbf{Q}_1=\mathbf{I}_{TN}-\bm{\mathcal{G}}_1( {\bm{\mathcal{G}}_1}^\dag\bm{\mathcal{G}}_1)^{-1}{\bm{\mathcal{G}}_1}^\dag,
    \end{eqnarray}
where $\mathbf{Q}_1$ is the projection matrix that projects a vector in $\mathbb{C}$ on to the orthogonal complementary subspace of $\mathbb{V}_2$ spanned by the column vectors of $\bm{\mathcal{G}}_1$, i.e.,
\begin{eqnarray}\label{v2}
 \mathbb{V}_2={\rm span} \{\mathcal{G}_1^1, \ \mathcal{G}_1^2, \ \ldots, \ \mathcal{G}_1^L \},
    \end{eqnarray}
    where   $\mathcal{G}_1^l$ is the $l$th column of $\bm{\mathcal{G}}_1$ for $l=1,2,\ldots, L$.

 Then,   multiplying $\mathbf{y}_1$ by $\mathbf{Q}_1$ to the left, we obtain that
    \begin{eqnarray}\label{Yc} \mathbf{Q}_1\mathbf{y}_1=\sqrt\frac{\rho}{\mu}\mathbf{Q}_1\bm{\mathcal{H}}_1 {\mathbf{s}} +\mathbf{Q}_1\mathbf{n}_1,
    \end{eqnarray}
such that the interference from user 2 is completely removed, i.e., $\mathbf{Q}_1\bm{\mathcal{G}}_1=\mathbf{0}$.

After   interference cancellation,  ML decoding is utilized to decode $ {\mathbf{s}}$ of user 1 at R1, and the  decision  metric is given by
 \begin{eqnarray}\label{icml}
 \hat{ {\mathbf{s}}} =\arg\min_{\mathbf{s} \in\mathcal{A}_1^{L}}\left\|\mathbf{Q}_1\mathbf{y}_{1}-\sqrt{\frac{\rho}{\mu}}\mathbf{Q}_{1}  \bm{\mathcal{H}}_1 { \mathbf{s}} \right\|^2.
    \end{eqnarray}
The ML metric can be used in (\ref{icml}) since the  noise term $\mathbf{Q}_1\mathbf{n}_1$ after interference cancellation is proved to be a degenerated Gaussian white noise in  \cite[\emph{Lemma 2}]{guo}.

\begin{definition} \label{def1}
\emph{(Full diversity in two-user  MIMO interference channel)}
Consider a transmission scheme for two-user MIMO interference channel in (\ref{Y1ice}), where each user has $M$ transmit antennas and its respective receiver has $N$ receive antennas.  User 1 is said to obtain full diversity of order $MN$ if its average pairwise error probability (PEP) $\mathcal{P}(\Delta \mathbf{s})$ decays like  $\rho^{-MN}$, i.e.,
\begin{eqnarray}\label{fulpep}
\mathcal{P}(\Delta \mathbf{s})\leq c \cdot \rho^{-MN},
\end{eqnarray}
where $c$ is a constant   and an error vector is defined as $\Delta  {\mathbf{s}}= {\mathbf{s}}-\bar{ {\mathbf{s}}}$ and $ {\mathbf{s}}$ is erroneously decoded as $\bar{ {\mathbf{s}}}$.
From (\ref{fulpep}), the full diversity  of  user 1 is defined  when  user 1  only considers the link between the corresponding receiver and itself.
\end{definition}

\begin{remark} [Difference from Guo-Xia's work\cite{guo}\cite{guocor}]\label{re1}
In \cite{guo}\cite{guocor},   the point-to-point MIMO systems are considered.   From the equivalent channel model $\mathbf{y}=\bm{\mathcal{H}}\mathbf{s}+\mathbf{n}$, the received signal is  divided into $K$ groups which results in
\begin{eqnarray}
\nonumber \mathbf{y}= {\bm{\mathcal{H}}_{1}} \mathbf{s}_{1}+{\bm{\mathcal{H}}_{2}} \mathbf{s}_{2}+\ldots+{\bm{\mathcal{H}}_{K}} \mathbf{s}_{K}+\mathbf{n},
 \end{eqnarray}
where $\mathbf{s}_k$ denotes the   symbol vector in the $k$th group corresponding to the channel matrix $\bm{\mathcal{H}}_k$ for $k=1,2,\ldots,K$.
  Suppose we want to decode $\mathbf{s}_1$.  The group ZF receiver is used to cancel the interference from all the other groups $\{\mathbf{s}_{2} \ \mathbf{s}_{3} \ \ldots \ \mathbf{s}_{K} \}$.  Note that the channel matrix $\bm{\mathcal{H}}_{1}$ in
 the desired group  and the channel matrices $\{\bm{\mathcal{H}}_{2} \ \bm{\mathcal{H}}_{3} \ \ldots \ \bm{\mathcal{H}}_{K} \}$ in the interfering groups  may be dependent on each other.


 In our work,  the interference from a different and independent transmitter (i.e., user 2) is eliminated by the group ZF receiver, thereby the channel matrices $\bm{\mathcal{G}}_{1}$ and $\bm{\mathcal{H}}_{1}$ are independent of each other.
\end{remark}

\begin{remark}[Independence between $\mathbf{Q}_1$ and $\bm{\mathcal{H}}_1$]\label{re2}
In Guo-Xia's work \cite{guo}\cite{guocor}, the cancellation matrix $\mathbf{Q}_1$ is constructed by   the channel matrices   $\{\bm{\mathcal{H}}_{2} \ \bm{\mathcal{H}}_{3} \ \ldots \ \bm{\mathcal{H}}_{K} \}$  in the interfering groups, which may be dependent on the desired group $\bm{\mathcal{H}}_1$ since all the elements in $\mathbf{Q}_1$  and $\bm{\mathcal{H}}_1$ are the linear combinations of $h_{i,j}$  and $h^{*}_{i,j}$ with $i=1,2,\ldots, M$ and  $j=1,2,\ldots, N$.

In our work, $\mathbf{Q}_1$ and $\bm{\mathcal{H}}_1$ in (\ref{Yc})  are independent, since $\mathbf{Q}_1$ in (\ref{Qcan}) is  derived from   $\bm{\mathcal{G}}_1$ and  $\bm{\mathcal{G}}_1$ is independent of $\bm{\mathcal{H}}_1$.
\end{remark}

\section{Code Design Criterion }
In this section, we propose a   design criterion for user 1 to achieve full diversity   with the group ZF receiver. Firstly, let us introduce the following lemma.


\begin{lemma}\label{thpic}

Consider   a transmission scheme for two-user MIMO interference channel   in  (\ref{Y1ice}). By using the group ZF receiver in (\ref{Qcan}), user 1 achieves  full diversity of order $MN$   if the following inequality can be satisfied,
\begin{eqnarray}\label{QHpic}
\| \mathbf{Q}_1\bm{\mathcal{H}}_1  \Delta  {\mathbf{s}} \|^2 \geq \alpha \sum^M_{i=1}\sum^N_{j=1} |\{h_1\}_{i,j}|^2,\  \forall \Delta   {\mathbf{s}} \in \Delta \mathcal{A}_1^L,
    \end{eqnarray}
 where  an error vector is defined as $\Delta  {\mathbf{s}}= {\mathbf{s}}-\bar{ {\mathbf{s}}}$ and $ {\mathbf{s}}$ is erroneously decoded as $\bar{ {\mathbf{s}}}$. Moreover,   $\{h_r\}_{i,j}$ and $\{g_r\}_{i,j}$  with $i=1,2,\ldots, M$ and  $j=1,2,\ldots, N$   are the $\{i,j\}$th entry of $\mathbf{H}_r$ and $\mathbf{G}_r$ for $r=1,2$, respectively.  $\alpha$ is a positive constant independent of $\{h_1\}_{i,j}$.
\end{lemma}

\begin{proof}

For convenience, $\{h_1\}_{i,j}$   and $\{g_1\}_{i,j}$  are simplified to be $h_{i,j}$ and $g_{i,j}$, respectively. In addition, we let $\mathbf{h}_j=\left[h_{1,j}, h_{2,j}, \ldots, h_{M,j}\right]^\mathcal{T}$ and $\mathbf{g}_j=\left[g_{1,j}, g_{2,j}, \ldots, g_{M,j}\right]^\mathcal{T}$.  Moreover, we define $\mathbf{h}=\left[\mathbf{h}^\mathcal{T}_1 \ \mathbf{h}^\mathcal{T}_2 \ \ldots \ \mathbf{h}^\mathcal{T}_N\right]^\mathcal{T}$ and   $\mathbf{g}=\left[\mathbf{g}^\mathcal{T}_1 \ \mathbf{g}^\mathcal{T}_2 \ \ldots \ \mathbf{g}^\mathcal{T}_N\right]^\mathcal{T}$.

Firstly, we evaluate the conditional  PEP $\mathcal{P}_{\{\mathbf{h},\mathbf{g}\}}( \Delta {\mathbf{s}} )$ of user 1 for two given  vectors $\mathbf{h}$ and $\mathbf{g}$. From (\ref{QHpic}), we obtain that
 \begin{eqnarray}\label{cpepth1}
\nonumber \mathcal{P}_{\{\mathbf{h},\mathbf{g}\}}( \Delta  {\mathbf{s}} )&=& Q\left(\sqrt{\frac{\rho}{\mu}\|\mathbf{Q}_1\bm{\mathcal{H}}_1  \Delta {\mathbf{s}}\|^2} \right)\\
\nonumber &\leq&  Q\left(\sqrt {{\alpha}\frac{\rho} {\mu} \sum^M_{i=1}\sum^N_{j=1} |h_{i,j}|^2 }\right)\\
 &\leq&\frac{1}{2}\exp\left(-\frac{\alpha\rho}{2\mu} \sum^M_{i=1}\sum^N_{j=1} |h_{i,j}|^2 \right),
  \end{eqnarray}
 where $Q(\cdot)$ denotes   $Q$-function and the last inequality is obtained since $Q(x)\leq \frac{1}{2}\exp \left(-\frac{x^2}{2}\right)$.

  By taking average over $\mathbf{h}$ in (\ref{cpepth1}),  we further have that
 \begin{eqnarray}
\nonumber \mathcal{P}_{\{ \mathbf{g}\}}  ( \Delta  {\mathbf{s}})&=& E_{\mathbf{h}}\{ \mathcal{P}_{\{\mathbf{h},\mathbf{g}\}} \left( \Delta {\mathbf{s}} \right)\}\\
 \nonumber &\leq& E_{\mathbf{h}}\bigg\{\frac{1}{2}\exp\left(-\frac{\alpha\rho}{2\mu} \sum^M_{i=1}\sum^N_{j=1} |h_{i,j}|^2 \right)\bigg\}\\
  \label{apep} &=& \frac{1}{2} \left( \frac{2\mu}{2\mu+ \alpha \rho}   \right)^{MN},
  \end{eqnarray}
 where the last equality holds since
  \begin{eqnarray}
\nonumber  E_{h} \{   \exp(-a |h|^2)   \}= \frac{1}{1+a}, \ h\sim \mathcal{CN}(0,1)  \ {\rm and} \ a> 0.
  \end{eqnarray}

Furthermore, when $\rho$ approaches infinite, the conditional PEP of user 1 for the given $\mathbf{g}$ can be upper bounded by
 \begin{eqnarray} \label{apep1th1}
  \mathcal{P}_{\{\mathbf{g}\}}( \Delta  {\mathbf{s}} )
  &\leq&  \frac{2^{MN-1}\mu^{MN}}{\alpha^{MN}  }   \rho^{-MN}.
  \end{eqnarray}

 Next, we take expectation over $\mathbf{g}$ in (\ref{apep1th1}), which has no impact on the power of $\rho$ (see \emph{Remark \ref{re1}}). Therefore,   the average PEP   $\mathcal{P} (  \Delta  {\mathbf{s}} )$ is upper bounded by
 \begin{eqnarray} \label{apep2th1}
\nonumber \mathcal{P} ( \Delta  {\mathbf{s}})
  &\leq&  \sigma  \rho^{-MN},
  \end{eqnarray}
 where $\sigma$ is a positive constant independent of $\mathbf{h}$. According to \emph{Definition \ref{def1}}, we can conclude that user 1 achieves full diversity of order $MN$  after interference cancellation.


Consequently, we  complete the proof of \emph{Lemma 1}.
\end{proof}

In what follows, we propose a design criterion for user 1 to satisfy the inequality  (\ref{QHpic}) with an  STBC  in two-user MIMO interference channels with the group ZF receiver.

\begin{theorem}\label{cri1}
Consider a transmission scheme for two-user MIMO interference channel in (\ref{Y1ice}) where each user has $M$ transmit antennas and the receiver has $N$ receive antennas. At the receiver R1, the group ZF receiver is used  to cancel the interference from user 2 with the STBC $\mathbf{C}$. After interference cancellation of $\mathbf{Q}_1\in \mathbb{C}^{TN \times TN}$ in (\ref{Qcan}),  user 1  with the STBC  $\mathbf{S}$ can satisfy the inequality  (\ref{QHpic}) in the following system
 \begin{eqnarray}\label{Yc1cri} \mathbf{Q}_1\mathbf{y}_1=\sqrt\frac{\rho}{\mu}\mathbf{Q}_1\bm{\mathcal{H}}_1 {\mathbf{s}} +\mathbf{Q}_1\mathbf{n}_1,
    \end{eqnarray}
if and only if \emph{ the matrix $\mathbf{Q}_1 \left(\mathbf{I}_N \otimes\Delta \mathbf{S}\right)$   has the full column rank of $MN$},
where  the matrix $\Delta \mathbf{S}\in \mathbb{C}^{T\times M}$ shares the same structure as the code matrix $\mathbf{S}$ and every entry is selected from   the erroneous vector $\Delta \mathbf{s}$.
\end{theorem}

\begin{proof}

 1) \emph{Sufficiency:} Firstly, we prove that the \emph{full-rank} condition    is  sufficient  to ensure the inequality (\ref{QHpic}) and thereby user 1 can guarantee full diversity after interference cancellation.
From (\ref{Y1ic}) and (\ref{Y1ice}), it is not hard to obtain that
\begin{eqnarray}\label{inqh1}
\|\mathbf{Q}_1\bm{\mathcal{H}}_1\Delta \mathbf{s}\|^2=\|\mathbf{Q}_1 \left(\mathbf{I}_N \otimes\Delta \mathbf{S}\right) \mathbf{h}\|^2,
\end{eqnarray}
where $\Delta \mathbf{s}$  corresponds to a code matrix $\Delta \mathbf{S}$ and a column vector $\mathbf{h}=\left[\mathbf{h}^\mathcal{T}_1 \ \mathbf{h}^\mathcal{T}_2 \ \ldots \ \mathbf{h}^\mathcal{T}_N\right]^\mathcal{T}$. Both $\Delta \mathbf{s}$ and $\mathbf{h}$ have been defined in \emph{Lemma \ref{thpic}}.

By using singular value decomposition (SVD) of the matrix $\mathbf{Q}_1 \left(\mathbf{I}_N \otimes\Delta \mathbf{S}\right)$, we further obtain that
\begin{eqnarray}\label{mlam}
 \|\mathbf{Q}_1\left(\mathbf{I}_N \otimes\Delta \mathbf{S}\right)\mathbf{h}\|^2&=& {\rm tr}\{\mathbf{h}^\dag\mathbf{U}^\dag\mathbf{\Lambda} \mathbf{U}\mathbf{h}\},
  \end{eqnarray}
where  $\mathbf{\Lambda}={\rm diag}\big\{\left[\lambda_1, \ \lambda_2, \ \ldots,  \lambda_{MN}\right]\big\}$ and $\mathbf{U}$ is a unitary matrix.
Since the matrix $\mathbf{Q}_1\left(\mathbf{I}_N \otimes\Delta \mathbf{S}\right)$   has the full column rank of $MN$, all the singular values $\lambda_i$
are positive.
Let $\lambda_{\min}=\min\big\{\lambda_1,  \lambda_2,  \ldots,  \lambda_{MN} \big\}$
and
$\alpha=\min_{\Delta \mathbf{s}\in \Delta \mathbf{S}}\lambda_{\min}$.
Since the signal constellation $\mathbf{S}$ is finite, the erroneous
set $\Delta \mathbf{S}$ is   finite as well. Since every $\lambda_{\min}>0$,
we conclude that $\alpha>0$.
Thus, $\|\mathbf{Q}_1\left(\mathbf{I}_N \otimes\Delta \mathbf{S}\right)\mathbf{h}\|^2$ can be lower bounded by
\begin{eqnarray}\label{fpepcri}
\nonumber \|\mathbf{Q}_1 \left(\mathbf{I}_N \otimes\Delta \mathbf{S}\right)\mathbf{h}\|^2&\geq& {\rm tr}\big\{\lambda_{{\min}}   (\mathbf{h}^\dag\mathbf{U}^\dag\mathbf{U} \mathbf{h})\big\}\\
 &\geq &  \alpha \sum^{M}_{i=1}\sum^{N}_{j=1}|h_{i,j}|^2,
  \end{eqnarray}
where the second inequality holds since $\mathbf{U}$ is a unitary matrix
and $\alpha$ is the minimum among all $\lambda_{\min}$.
Since   the matrix $\mathbf{Q}_1\left(\mathbf{I}_N \otimes\Delta \mathbf{S}\right)$ is independent of the channel $\bm{\mathcal{H}}_1$ (see \emph{Remark \ref{re2}}) and therefore independent of
the channel $\mathbf{h}$,  we obtain that $\alpha$ is independent of the channel $\mathbf{h}$.
This proves that   user 1  can fulfill the inequality (\ref{QHpic}) under  the group ZF receiver if the \emph{full-rank} condition is satisfied.


  2) \emph{Necessity:}   A proof by contradiction is considered to prove that the \emph{full-rank} condition   is necessary for user 1 to guarantee   (\ref{QHpic}) after the group ZF receiver. We assume that the matrix $\mathbf{Q}_1\left(\mathbf{I}_N \otimes\Delta \mathbf{S}\right)$ is  rank deficient,
     which means that it must exist a nonzero vector $\mathbf{a}=[a_1, a_2, \ldots, a_{MN}]$ with $a_v\in \mathbb{C}$ for $v=1, 2,\ldots, MN$ such that
      \begin{eqnarray}\label{qsv2cri}
a_1 \mathbf{Q}_{1} \Delta  S^{'}_1+a_2 \mathbf{Q}_{1} \Delta  S^{'}_2+\ldots +a_{MN} \mathbf{Q}_{1}\Delta  S^{'}_{MN} = \mathbf{0},
    \end{eqnarray}
    where $\Delta S^{'}_v$ stands for the $v$th column vector of the  matrix $\mathbf{I}_N \otimes\Delta \mathbf{S}$.
  Consequently, it is possible to find  a nonzero $\mathbf{h}=\left[a_1, a_2, \ldots, a_{MN}\right]^\mathcal{T}$ such that  we obtain $\|\mathbf{Q}_1\left(\mathbf{I}_N \otimes\Delta \mathbf{S}\right)\mathbf{h}\|^2=0$, which does not satisfy the inequality (\ref{QHpic}) for any positive constant $\alpha$. Thus,  the assumption contradicts with the condition that  user 1  with the   STBC  $\mathbf{S}$ can fulfill the inequality (\ref{QHpic}), thereby the \emph{full-rank} condition must hold.

  The proof of \emph{Theorem \ref{cri1}} is completed. 
\end{proof}

 By following \emph{Lemma \ref{thpic}},  we  conclude that  user 1 can achieve full diversity of order $MN$ under  the group ZF receiver if the \emph{full-rank} condition in \emph{Theorem \ref{cri1}} is satisfied.

\section{Proposed STBC Design }

 In this section, a systematic STBC design  is proposed for each user and two code examples are presented. Based on the proposed design criterion, we prove that the full diversity can be guaranteed with  our proposed STBC designs.

 \subsection{  STBC Designs}

The STBC designs  $\mathbf{S}$ and $\mathbf{C}$   for user 1 and user 2, respectively, are as follows
 \begin{eqnarray}
 \label{u1m}\mathbf{S}
 =\left[\begin{array}{ccccc}
\tilde{s}_1  &   0 & \ldots  & 0\\
\tilde{s}_2 &  \tilde{s}_1 & \ddots   & \vdots\\
\vdots &  \tilde{s}_2  &  \ddots&   0\\
\tilde{s}_{L} &  \ddots&  \ddots &   \tilde{s}_1 \\
0 & \tilde{s}_{L} &  \ddots &  \tilde{s}_2 \\
\vdots & \ddots&  \ddots  & \vdots\\
0 &  \ldots&  0  &   \tilde{s}_{L}  \\
&&\mathbf{0}_{M\times M}
 \end{array}\right]_{(L+2M-1 )\times M}\\
 {\rm and}\
\mathbf{C}
=\left[\begin{array}{ccccc}
&&\mathbf{0}_{M\times M}\\
\tilde{c}_1  &   0 & \ldots  & 0\\
\tilde{c}_2 &  \tilde{c}_1 & \ddots   & \vdots\\
\vdots &  \tilde{c}_2  &  \ddots&   0\\
\tilde{c}_{L} &  \ddots&  \ddots &   \tilde{c}_1 \\
0 & \tilde{c}_{L} &  \ddots &  \tilde{c}_2 \\
\vdots & \ddots&  \ddots  & \vdots\\
0 &  \ldots&  0  &   \tilde{c}_{L}  \\
 \end{array}\right]_{(L+2M-1 )\times M},
 \end{eqnarray}
where $L$ stands for the number of layers in each codeword. Moreover, $\tilde{s}_l$ and $\tilde{c}_l$ are the $l$th elements of the rotated symbol vectors  $\tilde{\mathbf{s}}$ and $\tilde{\mathbf{c}}$, respectively, which are given by
\begin{eqnarray}\label{rosc1}
 \tilde{\mathbf{s}}&=&\mathbf{\Theta}\mathbf{s}=\left[ \tilde{s}_1, \ \tilde{s}_2, \ldots, \ \tilde{s}_L\right]^\mathcal{T}\\
 \label{rosc2} \ {\rm and}\  \tilde{\mathbf{c}}&=&\mathbf{\Theta}\mathbf{c}=\left[ \tilde{c}_1, \ \tilde{c}_2, \ldots, \ \tilde{c}_L\right]^\mathcal{T},
  \end{eqnarray}
where the rotation matrix $\mathbf{\Theta}$  is given in \cite{fd} and \cite{WangGY}    such that both $\tilde{\mathbf{s}}$ and $\tilde{\mathbf{c}}$   have non-zero product distance, i.e., full rank if they are used to form the diagonal space-time codes \cite{dast}. Similar to \cite{liu2}, the proposed codes are featured as Toeplitz coding structure. However, in order to achieve full diversity after interference cancellation, the independent symbols in each layer of the proposed codes are transformed by a rotation matrix $\mathbf{\Theta}$, which will be specified in the following.


The code rate of an STBC is defined as $\mathcal{R}=\frac{L}{T}$ symbols per channel use, where $L$  independent symbols selected from a finite constellation in the STBC are transmitted over $T$ time slots. Accordingly,
 the code rate of the proposed STBC for each user is $\frac{L}{L+2M-1}$ symbols per channel use with $T=L+2M-1$. For a fixed $M$, the rate approaches $1$ symbol per channel use when   $L$ in the time delay in the encoding (or code block size) goes to infinity.

\begin{remark}[Difference from the multilayer STBC in \cite{wei}]\label{re3}
The multilayer  STBC  $\mathbf{X}$ in \cite[Eq. (13)]{wei} was proposed to achieve full diversity for MIMO systems based on the criterion in \cite{guo}\cite{guocor}. In \cite{wei}, the STBC design  is given by
\begin{eqnarray}\label{it12}
\mathbf{X}= \left[\begin{array}{cccc}
    X_{1,1}  & 0              & \cdots      & 0 \\
    X_{2,1}    &  X_{1,2}  & \ddots      & \vdots \\
     \vdots                &  X_{2,2}    & \ddots      & 0 \\
    X_{L,1}      &  \vdots                 & \ddots      &  X_{1,M} \\
     0           & X_{L,2}       & \ddots      & X_{2,M}\\
    \vdots                 & \ddots                   &\ddots       & \vdots\\
    0             &  \cdots                 &0       & X_{L,M} \\
     \end{array}
  \right],
\end{eqnarray}
where the symbol vector $\left[{X}_{l,1},  \   {X}_{l,1}, \ \ldots,   \ {X}_{l,M}\right]^\mathcal{T}$ in the $l$th layer is given by
  \begin{eqnarray}
\left[\begin{array}{cccc}
 {X}_{l,1},  &   {X}_{l,1}, & \ldots,  & {X}_{l,M}
 \end{array}\right]^\mathcal{T}= \mathbf{\Theta}\mathbf{s}_l,
  \end{eqnarray}
and the transmitted symbol vector $\mathbf{s}_l$ is given by
\begin{eqnarray}\label{it12s}
 \mathbf{s}_l=\left[\begin{array}{cccc}
 {s}_{(l-1)M+1},  &    {s}_{(l-1)M+2}, & \ldots,  &  {s}_{lM}
 \end{array}\right]^\mathcal{T},
  \end{eqnarray}
  for $l=1,2,\ldots, L$.

Similar to our STBC designs in (\ref{u1m}), the   design in (\ref{it12}) has $L$ layers and  the transmitted symbols in each layer is performed by a rotation matrix  $\mathbf{\Theta}$. Different from   (\ref{u1m}),   the code structure of (\ref{it12}) is not featured as a Toeplitz matrix since  the  entries in each layer  are different. With the code design in (\ref{it12}), user 1 may not obtain the full diversity after interference cancellation since \emph{full-rank} condition in \emph{Theorem \ref{cri1}} cannot be satisfied.

\end{remark}

\subsection{  Full Diversity Property}

Now, we   show that each user   can achieve full diversity after interference cancellation with  the proposed STBCs. Without loss of generality, we focus on user 1 only.

\begin{theorem}\label{thpro}
Consider two-user MIMO interference channels where each user has $M$ transmit antennas and each receiver is equipped with $N$ receive antennas.  The  STBC designs in (\ref{u1m}) can obtain full diversity of  order $MN$ under the  group ZF receiver .
\end{theorem}

 \begin{proof}

 To prove that   user 1 with the proposed STBC $\mathbf{S}$ in (\ref{u1m})   obtains full diversity,
 we   show that $\mathbf{Q}_1\left(\mathbf{I}_N \otimes\Delta \mathbf{S}\right)$ has full  rank of $MN$ by following \emph{Lemma \ref{thpic}} and \emph{Theorem \ref{cri1}}.

 From the code structure in (\ref{u1m}), we have that
 \begin{align}\label{newS}
&\nonumber \Delta \mathbf{S}= \left[ \Delta S_1 \ \Delta S_2 \ \ldots \ \Delta S_M \right]\\
&=\left[\begin{array}{ccccc}
\Delta\tilde{s}_1  &   0 & \ldots  & 0\\
\Delta\tilde{s}_2 & \Delta \tilde{s}_1 & \ddots   & \vdots\\
\vdots & \Delta \tilde{s}_2  &  \ddots&   0\\
\Delta \tilde{s}_{L} &  \ddots&  \ddots &   \Delta \tilde{s}_1 \\
0 & \Delta \tilde{s}_{L} &  \ddots &  \Delta \tilde{s}_2 \\
\vdots & \ddots&  \ddots  & \vdots\\
0 &  \ldots&  0  &  \Delta \tilde{s}_{L}  \\
&&\mathbf{0}_{M\times M}
 \end{array}\right]_{(L+2M-1 )\times M},
  \end{align}
where $\Delta S_i $ is the $i$th column vector of $\Delta \mathbf{S}$  for $i=1,2,\ldots, M$  and
 \begin{eqnarray}\label{dds}
 \Delta \tilde{\mathbf{s}}=\mathbf{\Theta}\Delta \mathbf{s}=\left[\Delta \tilde{s}_1, \ \Delta \tilde{s}_2, \ \ldots, \Delta \tilde{s}_L \right]^\mathcal{T},
 \end{eqnarray}
where $\Delta \tilde{{s}}_l \neq 0, 1\leq l\leq L$, from the design of the rotation matrix $\mathbf{\Theta}$ in \cite{fd} and \cite{WangGY}. Generally, the equivalent channel matrix   of a linear dispersion STBC is   expressed as a stack version of the equivalent channel matrices of all individual receive antennas\cite{has}. Hence, the equivalent channel matrix $\bm{\mathcal{G}}_1$ is given by
\begin{eqnarray}\label{geq1}
\bm{\mathcal{G}}_1=\left[\bm{\mathcal{G}}^{\mathcal{T}}_{1,1} \ \bm{\mathcal{G}}^{\mathcal{T}}_{1,2} \ \ldots \bm{\mathcal{G}}^{\mathcal{T}}_{1,N}\right]^{\mathcal{T}},
\end{eqnarray}
where $\bm{\mathcal{G}}_{1,j}$ denotes the equivalent channel matrix of user 2 corresponding to the $j$th receive antenna. From  (\ref{Y1ice}), (\ref{u1m}) and (\ref{rosc2}), we have
  \begin{eqnarray}\label{theeq}
\bm{\mathcal{G}}_{1,j} =  \bm{\underline{\mathcal{G}}}_{1,j} \mathbf{\Theta},
\end{eqnarray}
with
  \begin{eqnarray}
 \label{eqLgg}   \nonumber &&\underline{\bm{\mathcal{G}}}_{1,j}=\left[\underline{\mathcal{G}}^1_{1,j} \ \underline{\mathcal{G}}_{1,j}^2 \ \ldots \ \underline{\mathcal{G}}_{1,j}^{L} \right]\\
 &&\nonumber=\left[\begin{array}{cccc }
    &&\mathbf{0}_{M\times L}\\
 {g}_{1,j}  &   0 & \ldots  & 0\\
 {g}_{2,j} &   {g}_{1,j} & \ddots   & \vdots\\
\vdots &   {g}_{2,j}  &  \ddots&   0\\
 {g}_{M,j} &  \ddots&  \ddots &    {g}_{1,j} \\
0 &  {g}_{M,j} &  \ddots &   {g}_{2,j} \\
\vdots & \ddots&  \ddots  & \vdots\\
0 &  \ldots&  0  &    {g}_{M,j}  \\
 \end{array}\right]_{(L+2M-1 )\times L},
    \end{eqnarray}
   where $\underline{\mathcal{G}}_{1,j}^l$  denotes the $l$th column vector of $\bm{\underline{\mathcal{G}}}_{1,j}$ for $l=1,2,\ldots, L$. By substituting (\ref{theeq})  into (\ref{geq1}), the equivalent channel $\bm{\mathcal{G}}_{1}$ is rewritten as
      \begin{eqnarray}\label{theeq1}
\bm{\mathcal{G}}_{1} =  \left[ \underline{\bm{\mathcal{G}}}_{1,1}^\mathcal{T} \ \underline{\bm{\mathcal{G}}}_{1,2}^\mathcal{T} \ \ldots \underline{\bm{\mathcal{G}}}_{1,N}^\mathcal{T}\right]^\mathcal{T} \mathbf{\Theta},
\end{eqnarray}
where we let $\bm{\underline{\mathcal{G}}}_{1}=\left[ \underline{\bm{\mathcal{G}}}_{1,1}^\mathcal{T} \ \underline{\bm{\mathcal{G}}}_{1,2}^\mathcal{T} \ \ldots \underline{\bm{\mathcal{G}}}_{1,N}^\mathcal{T}\right]^\mathcal{T}$.

From (\ref{Qcan}), the cancellation matrix $\mathbf{Q}_{1}$   is given by
     \begin{eqnarray}
\nonumber \mathbf{Q}_{1}&=&\mathbf{I}_{TN}-\bm{{\mathcal{G}}}_{1}( {\bm{{\mathcal{G}}}_{1}}^\dag\bm{{\mathcal{G}}}_{1})^{-1}{{\bm{\mathcal{G}}}_{1}}^\dag,\\
 \label{newq}&=&\mathbf{I}_{TN}-\bm{\underline{\mathcal{G}}}_{1}( {\bm{\underline{\mathcal{G}}}_{1}}^\dag\bm{\underline{\mathcal{G}}}_{1})^{-1}{\underline{\bm{\mathcal{G}}}_{1}}^\dag.
    \end{eqnarray}
where the second equality  holds since   $\mathbf{\Theta}$ is a full rank matrix \cite{fd} and \cite{WangGY}. In this scenario, $\mathbf{Q}_{1}$  can be regarded as  a projection matrix that projects a vector in $\mathbb{C}$ onto the orthogonal complementary subspace of $\underline{\mathbb{V}}_2$. The subspace $\underline{ \mathbb{V}}_2$ is  spanned by all the column vectors $\underline{{\mathcal{G}}}^l_1$ of $\bm{\underline{\mathcal{G}}}_{1}$ for $l=1,2,\ldots,L$.

Then, we are ready to check whether the matrix $\mathbf{Q}_1 \left(\mathbf{I}_N \otimes\Delta \mathbf{S}\right)$  has full rank. For a clear explanation,  the following matrix is used to see the relationship between $\left(\mathbf{I}_N \otimes\Delta \mathbf{S}\right)$ and $\underline{\bm{\mathcal{G}}}_1$
\begin{eqnarray}\label{resg}
&&\nonumber \left[\mathbf{I}_N \otimes\Delta \mathbf{S} \  \  \underline{\bm{\mathcal{G}}}_1 \right]\\
\nonumber &=&\left[\begin{array}{cccccccc }
 \Delta{S}^{'}_1  &   \Delta{S}^{'}_2 & \ldots  & \Delta{S}^{'}_{MN}    &  \underline{ {\mathcal{G}}}^1_{1} & \underline{ {\mathcal{G}}}_{1}^2 & \ldots & \underline{ {\mathcal{G}}}_{1}^{L}       \\
 \end{array}\right]\\
&=&\left[\begin{array}{ccccc}
\Delta \mathbf{S}  &    \mathbf{0} & \ldots  & \mathbf{0}   & \underline{\bm{\mathcal{G}}}_{1,1} \\
 \mathbf{0} & \Delta \mathbf{S} & \ddots   & \vdots & \underline{\bm{\mathcal{G}}}_{1,2} \\
\vdots  & \ddots   &  \ddots&   \mathbf{0} & \vdots\\
 \mathbf{0} & \ldots   & \mathbf{0} &   \Delta \mathbf{S}& \underline{\bm{\mathcal{G}}}_{1,N}\\
 \end{array}\right],
 \end{eqnarray}
where $\Delta{S}^{'}_r$ denotes the $r$th column vector in the matrix $\mathbf{I}_N \otimes\Delta \mathbf{S}$ for $r=1,2,\ldots,MN$. We can observe that the vector groups $\left[\Delta{S}^{'}_{(j-1)M+1} \ \Delta{S}^{'}_{(j-1)M+2} \ \ldots \ \Delta{S}^{'}_{jM}\right]$ for $j=1,2,\ldots, N$   are orthogonal to each other. Motivated by the observation, in order to prove the full rank of the matrix $\mathbf{Q}_1 \left(\mathbf{I}_N \otimes\Delta \mathbf{S}\right)$, it is   equivalent to show  that a linear combination over $\mathbb{C}$ of all vectors in the $j$th group $\left[\Delta{S}^{'}_{(j-1)M+1} \ \Delta{S}^{'}_{(j-1)M+2} \ \ldots \ \Delta{S}^{'}_{jM}\right]$ with $\forall j\in\{1, 2, \ldots, N\}$ does not belong to the subspace $\underline{\mathbb{V}}_2$.

Without loss of generality, we  consider the $j$th vector group. It is noticed that all the rows in this group are all zeros except the $(j-1)(L+2M-1)+1$th
$\sim j(L+2M-1)$th rows. From  (\ref{newS}) and (\ref{eqLgg}), we have
\begin{align}
 &\nonumber \left[\mathbf{\Delta \mathbf{S}} \ | \ \underline{\bm{\mathcal{G}}}_{1,j} \right]\\
&=\nonumber  \left[\begin{array}{cccccccc }
 \Delta{S}_1  &   \Delta{S}_2& \ldots  & \Delta{S}_M    &  \underline{\mathcal{G}}^1_{1,j} & \underline{\mathcal{G}}_{1,j}^2 & \ldots & \underline{\mathcal{G}}_{1,j}^{L}       \\
 \end{array}\right]\\
 &=\label{dsg1}\left[\begin{array}{cccccccc }
 \Delta{\tilde{s}}_1   & &  &             \\
 \Delta{\tilde{s}}_2   & \ddots   &  & & \mathbf{0}_{M\times L}\\
\vdots       &  \ddots&   \Delta{\tilde{s}}_1\\
\Delta{\tilde{s}}_{M}   &  \ddots&  \Delta{\tilde{s}}_2& {g}_{1,j}    & &\\
  \vdots &   \ddots &     \vdots & {g}_{2,j}  & \ddots   &   \\
  \Delta{\tilde{s}}_{L} &    \ddots &    \Delta{\tilde{s}}_{M} & \vdots &  \ddots &   {g}_{1,j}\\
  &    \ddots  &  \vdots & {g}_{M,j} &     \ddots   &  {g}_{2,j}\\
  &      &     \Delta{\tilde{s}}_{L} &  &    \ddots&   \vdots \\
  &  \mathbf{0}_{M\times M}&&   &     &    {g}_{M,j}\\
  & &  &  &    &
 \end{array}\right]
    \end{align}
 where   $\Delta S_i$ denotes the $i$th column vector of $\mathbf{\Delta \mathbf{S}}$ for $i=1,2,\ldots,M$. Note that the linear combination of all the  column vectors in $\Delta{\mathbf{S}}$ does not belong to the subspace spanned by the column vectors in $\underline{\bm{\mathcal{G}}}_{1,j}$.
The fact   holds since  the entry $\Delta \tilde{s}_1$  is nonzero for any nonzero vector $\Delta \tilde{\mathbf{s}}$ in (\ref{dds}). The second equality in (\ref{dsg1}) is given when $L>M$ and  the fact can be easily obtained  when $L\leq M$.

Therefore, any nonzero linear combination over $\mathbb{C}$ of  the column vectors in  $\mathbf{I}_N \otimes\Delta \mathbf{S}$ does not belong to the subspace $ {\underline{\mathbb{V}}}_2$ spanned by the column vectors of $ \underline{{\bm{\mathcal{G}}}}_{1}$, i.e.,
\begin{eqnarray}\label{sv2}
a_1 \Delta  S^{'}_1+a_2 \Delta  S^{'}_2+\ldots +a_{MN} \Delta  S^{'}_{MN}  \nsubseteq \underline{\mathbb{V}}_2,
    \end{eqnarray}
for not all zero $ a_v \in \mathbb{C}, v=1,2,\ldots, MN$.
Then, it follows that
 \begin{eqnarray}\label{qsv2}
a_1 \mathbf{Q}_{1} \Delta  S^{'}_1+a_2 \mathbf{Q}_{1} \Delta  S^{'}_2+\ldots +a_{MN} \mathbf{Q}_{1}\Delta  S^{'}_{MN} \neq \mathbf{0}.
    \end{eqnarray}
Therefore,  all   column vectors of the matrix $\mathbf{Q}_1 \left(\mathbf{I}_N \otimes\Delta \mathbf{S}\right)$ are linear independent over $\mathbb{C}$, which justifies that the matrix has the full column rank of $MN$. By following \emph{Theorem \ref{cri1}}, user 1 with the proposed STBC $\mathbf{S}$ in (\ref{u1m}) retains the full diversity after interference cancellation.

  The proof of \emph{Theorem \ref{thpro}} is completed.

\end{proof}

\subsection{Code Examples}

1) \emph{$M=2$ and $L=4$}

   The proposed STBC designs     for $M=2$ and $L=4$ are given as follows
\begin{eqnarray}\label{m2l4}
\mathbf{S} = \left[\begin{array}{cc}
   \tilde{s}_1 & 0\\
    \tilde{s}_2 & \tilde{s}_1\\
    \tilde{s}_3 & \tilde{s}_2\\
    \tilde{s}_4 & \tilde{s}_3\\
            0   & \tilde{s}_4\\
   0 & 0\\
   0 & 0
  \end{array}
  \right] {\rm and}\
  \mathbf{C} = \left[\begin{array}{cc}
    0 & 0\\
   0 & 0\\
   \tilde{c}_1 & 0\\
    \tilde{c}_2 & \tilde{c}_1\\
    \tilde{c}_3 & \tilde{c}_2\\
    \tilde{c}_4 & \tilde{c}_3\\
            0   & \tilde{c}_4\\
  \end{array}
  \right].
\end{eqnarray}

According to (\ref{rosc1}) and  (\ref{rosc2}), the symbol vectors $\mathbf{\tilde{s}}$ and $\mathbf{\tilde{c}}$  are given by
\begin{eqnarray}\label{laym2l4}
  \nonumber \tilde{\mathbf{s}}  &=& [\tilde{s}_{1}, \  \tilde{s}_{2}, \ \tilde{s}_{3}, \ \tilde{s}_{4}]^\mathcal{T}
= \mathbf{\Theta} [s_{ 1},\ s_{ 2}, \ s_{ 3},\ s_{ 4}]^\mathcal{T}\\ \nonumber {\rm and} \
\tilde{\mathbf{c}}  &=& [\tilde{c}_{1}, \  \tilde{c}_{2}, \ \tilde{c}_{3}, \ \tilde{c}_{4}]^\mathcal{T}
= \mathbf{\Theta} [c_{ 1},\ c_{ 2},\ c_{ 3},\ c_{ 4}]^\mathcal{T},
\end{eqnarray}
where the matrix $\mathbf{\Theta}$ with size of $  {4\times 4}$ is selected from \cite{fd} and the code rate is  $4/7$ symbol per channel use.

2) \emph{$M=3$ and $L=4$}

The code examples     for $M=3$ and $L=4$ are given as follows
\begin{eqnarray}\label{m3l4}
\mathbf{S} = \left[\begin{array}{ccc}
   \tilde{s}_1 & 0 & 0\\
    \tilde{s}_2 & \tilde{s}_1& 0\\
    \tilde{s}_3 & \tilde{s}_2& \tilde{s}_1\\
    \tilde{s}_4 & \tilde{s}_3& \tilde{s}_2\\
            0   & \tilde{s}_4& \tilde{s}_3\\
   0 & 0& \tilde{s}_4\\
   0 & 0&0\\
   0 & 0&0\\
   0 & 0&0
  \end{array}
  \right] {\rm and}\
  \mathbf{C} = \left[\begin{array}{ccc}
    0 & 0&0\\
   0 & 0&0\\
   0 & 0&0\\
      \tilde{c}_1 & 0 & 0\\
    \tilde{c}_2 & \tilde{c}_1& 0\\
    \tilde{c}_3 & \tilde{c}_2& \tilde{c}_1\\
    \tilde{c}_4 & \tilde{c}_3& \tilde{c}_2\\
            0   & \tilde{c}_4& \tilde{c}_3\\
            0   & 0& \tilde{c}_4\\
  \end{array}
  \right].
\end{eqnarray}

The symbol vectors $\mathbf{\tilde{s}}$ and $\mathbf{\tilde{c}}$  are given by
\begin{eqnarray}\label{laym3l4}
  \nonumber \tilde{\mathbf{s}}  &=& [\tilde{s}_{1}, \  \tilde{s}_{2}, \ \tilde{s}_{3}, \ \tilde{s}_{4}]^\mathcal{T}
= \mathbf{\Theta} [s_{ 1},\ s_{ 2}, \ s_{ 3},\ s_{ 4}]^\mathcal{T}\\ \nonumber {\rm and} \
\tilde{\mathbf{c}}  &=& [\tilde{c}_{1}, \  \tilde{c}_{2}, \ \tilde{c}_{3}, \ \tilde{c}_{4}]^\mathcal{T}
= \mathbf{\Theta} [c_{ 1},\ c_{ 2},\ c_{ 3},\ c_{ 4}]^\mathcal{T},
\end{eqnarray}
where the matrix $\mathbf{\Theta}$ with size of $  {4\times 4}$ is selected from \cite{fd}. The code rate is  $4/9$ symbol per channel use.

\section{Simulation Results}

In this section, the average bit error rate (BER) of the proposed STBC in (\ref{u1m})  with the group ZF  receiver  is simulated. We consider a two-user interference channel model with $M$ transmit antennas at each user and $N$ receive antennas at each receiver. The maximum achievable diversity order for each user in this system is $MN$.

Fig. \ref{fig2} depicts  BER  performance of the   proposed codes with different number of transmit antennas and single receive antenna.
  $4$QAM is used for two users. The dashed lines with squares and dots  are plotted as the references with diversity order $2$ and $3$, respectively. $C_1$ and $C_2$ are two positive constants. As expected, all types of the proposed  STBCs can achieve full diversity of order $M$ under the group ZF receiver.  For a fixed $M$, the coding gain of our proposed STBC is degraded  with the increase of the number of layers, since more interference are produced by the symbols in the extra layers.

\begin{figure}
      \centering \scalebox{0.519} {\includegraphics{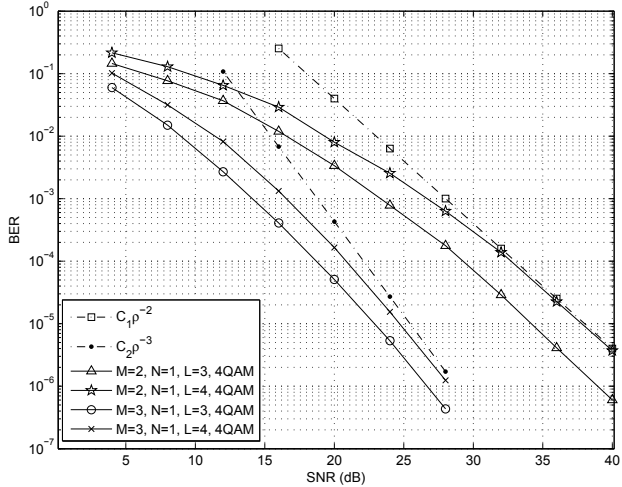}}
       \caption{Performance of the proposed codes for user 1 with $M$ transmit antennas, $1$ receive antenna and $L$ layers in two-user interference channels.}
        \label{fig2}
\end{figure}
%
%

Fig. \ref{fig4} presents BER performance comparison of user 1 with the proposed STBC when it is decoded by  the group ZF receiver  and the asymptotic optimal (AO) receiver (or MMSE receiver) \cite[\emph{Section V}]{guo}.   4QAM is used for both of the cancellation schemes. The dashed lines with squares and dots   are used as  the references with diversity order $2$ and $3$, respectively. $C_3$ and $C_4$ are two positive constants. It can be observed that user 1 with our proposed STBC is capable of  achieving full diversity of order $M$ under both group ZF and AO/MMSE receivers. Since the performance of  the AO receiver  always outperforms the group ZF receiver, the proposed codes  can also achieve the full diversity as evidenced from this simulation.
Moreover, when both $M$ and $L$ are fixed, the coding gain  with the AO receiver is about 1dB better than the one with the group ZF receiver.

\begin{figure}
        \centering \scalebox{0.522} {\includegraphics{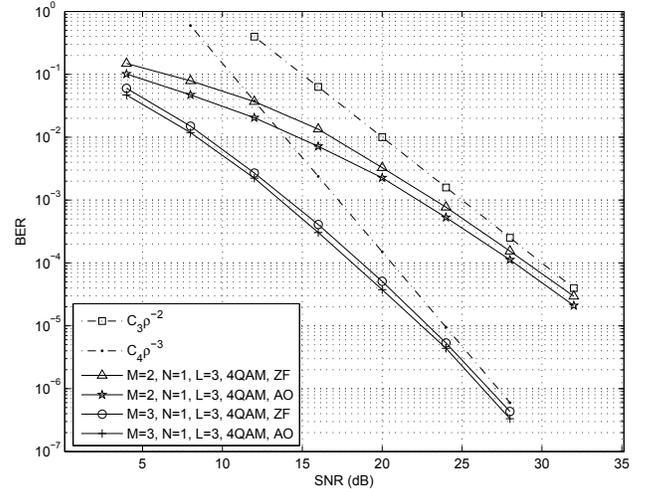}}
       \caption{Performance of the proposed codes for user 1 with the group ZF and AO receivers in two-user interference channels.}
        \label{fig4}
\end{figure}

\section{Conclusion}

In this paper,   a design criterion of STBC design was proposed for two-user MIMO interference channels under the group ZF receiver.  Based on the criterion, we proposed a systematic STBC for each user to obtain full diversity after interference cancellation. The proposed STBC approaches a code rate of  one symbol per channel use when the time delay in the encoding (or code block size) gets large.  Simulation results were presented to show that our proposed STBCs can achieve full diversity under the group ZF receiver.

\begin{biography}[{\includegraphics[width=1.1in,height=1.3in,clip,keepaspectratio]{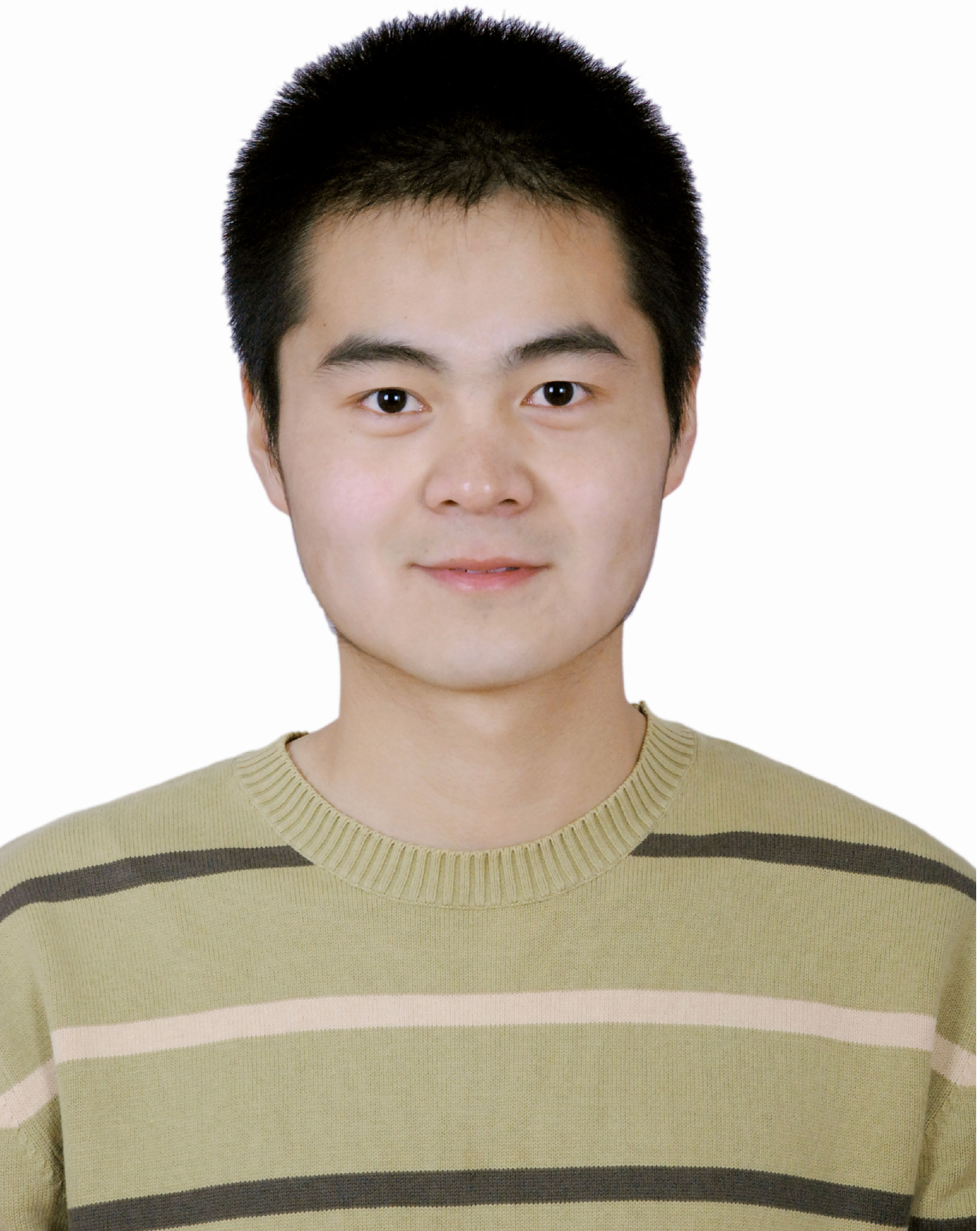}}]
{Long Shi} (S'10) received his B.S. degree in electrical
engineering from Changchun Institute of Technology, Jilin, China, in 2007 and M.S. degree in telecommunication engineering from Jilin University, Jilin, China,  in 2009.
He was a visiting
student at The Chinese University of Hong Kong and University of Delaware in 2010 and 2011, respectively. He is currently working toward his Ph.D. degree in electrical engineering at The
University of New South Wales, Sydney, Australia. His current research
interests include   space-time coding in MIMO systems, cooperative systems, and interference channels.
\end{biography}
\begin{biography}[{\includegraphics[width=1in,height=1.25in,clip,keepaspectratio]{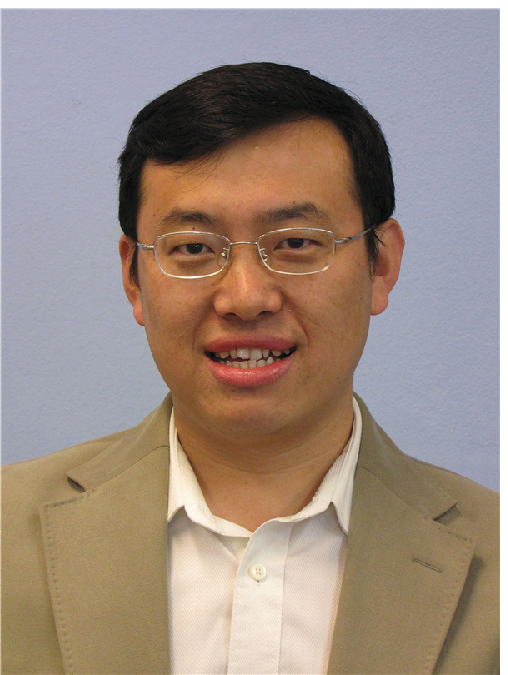}}]
{Wei Zhang} (S'01-M'06-SM'11) received the Ph.D. degree in electronic
engineering from The Chinese University of Hong Kong in 2005. He was
a Research Fellow with the Department of Electronic and Computer
Engineering, Hong Kong University of Science and Technology, during
2006-2007. From 2008, he has been with the School of Electrical
Engineering and Telecommunications, The University of New South
Wales, Sydney, Australia, where he is an Associate Professor. His current
research interests include cognitive radio, cooperative
communications, space-time coding, and multiuser MIMO.

He received the best paper award at the 50th IEEE Global
Communications Conference (GLOBECOM), Washington DC, in 2007 and the
IEEE Communications Society Asia-Pacific Outstanding Young
Researcher Award in 2009. He is Co-Chair of Communications Theory
Symposium of International Conference on Communications (ICC 2011),
Kyoto, Japan. He is an Editor of the IEEE Transactions on Wireless
Communications and an Editor of the IEEE Journal on Selected Areas
in Communications (Cognitive Radio Series).
\end{biography}
\begin{biography}[{\includegraphics[width=1.1in,height=1.3in,clip,keepaspectratio]{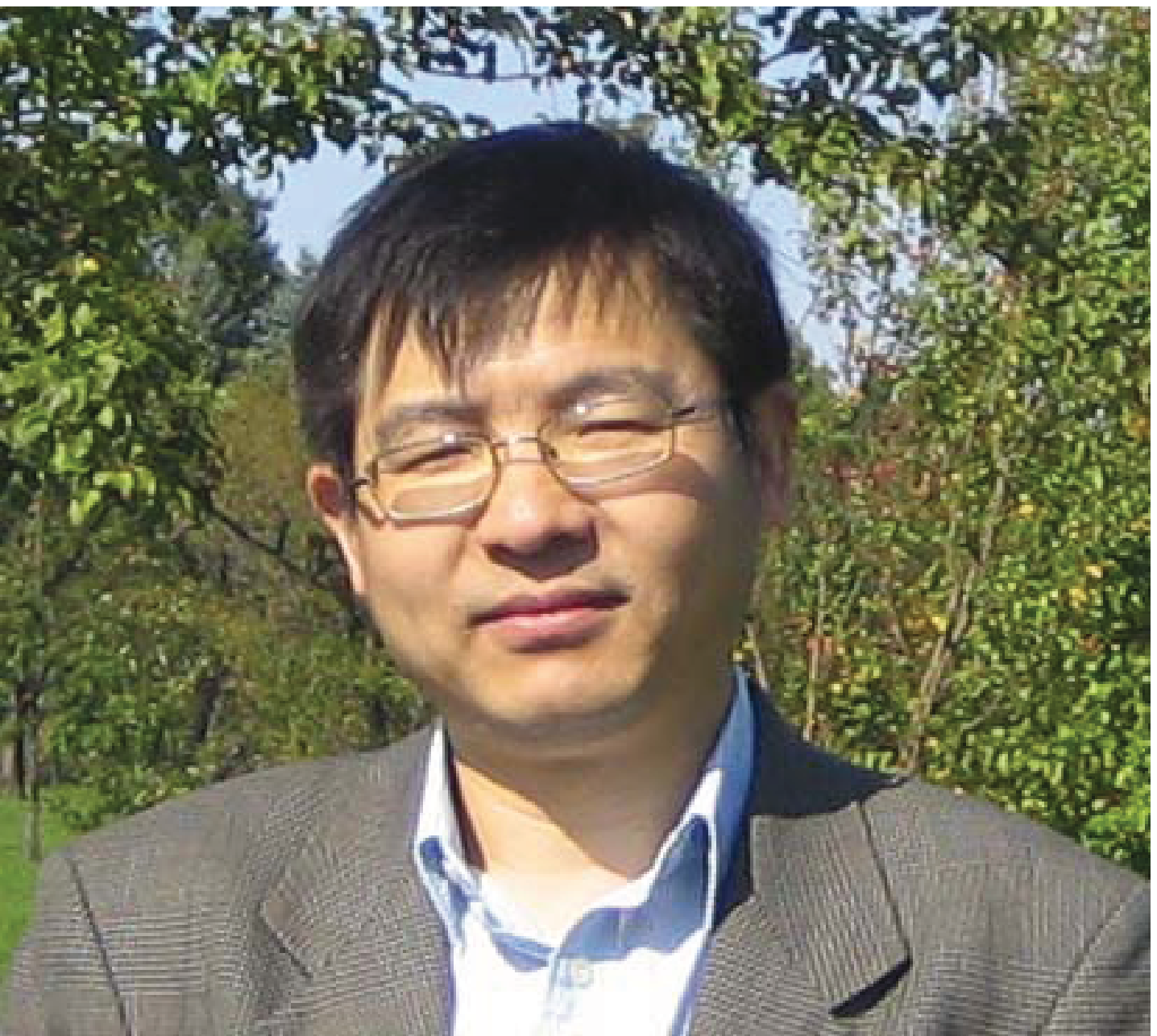}}]{Xiang-Gen Xia}
(M'97-SM'00-F'09) received his B.S. degree in mathematics from Nanjing Normal University, Nanjing, China, and his M.S. degree in mathematics from Nankai University, Tianjin, China, and his Ph.D. degree in Electrical Engineering from the University of Southern California, Los Angeles, in 1983, 1986, and 1992, respectively.

He was a Senior/Research Staff Member at Hughes Research Laboratories, Malibu, California, during 1995-1996. In September 1996, he  joined the Department of Electrical and Computer Engineering, University of Delaware, Newark, Delaware, where he is the Charles Black Evans Professor. His current research interests include space-time coding, MIMO and OFDM systems, digital signal processing, and SAR and ISAR imaging. Dr. Xia has  over 240 refereed journal articles published and accepted, and 7 U.S. patents awarded and is the author of the book  Modulated Coding for Intersymbol Interference Channels (New York, Marcel Dekker, 2000).

Dr. Xia received the National Science Foundation (NSF) Faculty Early Career Development (CAREER) Program Award in 1997, the Office of Naval Research (ONR) Young Investigator Award in 1998, and the Outstanding Overseas Young Investigator Award from the National Nature Science Foundation of China in 2001. He also received the Outstanding Junior Faculty Award of the Engineering School of the University of Delaware in 2001. He is currently an Associate Editor of the IEEE Transactions on Wireless Communications, IEEE Transactions on Signal Processing, Signal Processing (China), and the Journal of Communications and Networks (JCN). He was a guest editor of Space-Time Coding and Its Applications in the EURASIP Journal of Applied Signal Processing in 2002. He served as an Associate Editor of the IEEE Transactions on Signal Processing during 1996 to 2003, the IEEE Transactions on Mobile Computing during 2001 to 2004, IEEE Transactions on Vehicular Technology during 2005 to 2008, the IEEE Signal Processing Letters during 2003 to 2007, Signal Processing (EURASIP) during 2008 to 2011, and the EURASIP Journal of Applied Signal Processing during 2001 to 2004. Dr. Xia served as a Member of the Signal Processing for Communications Committee from 2000 to 2005 and  a Member of the Sensor Array and Multichannel (SAM) Technical Committee from 2004 to 2009 in the IEEE Signal Processing Society. He serves as IEEE Sensors Council Representative of IEEE Signal Processing Society (from 2002) and served as the Representative of IEEE Signal Processing Society to the Steering Committee for IEEE Transactions on Mobile Computing during 2005 to 2006. Dr. Xia is Technical Program Chair of the Signal Processing Symp., Globecom 2007 in Washington D.C. and  the General Co-Chair of ICASSP 2005 in Philadelphia.

\end{biography}

\end{document}